\documentclass[11pt]{article}
\usepackage[margin=1in]{geometry}
\usepackage[title]{appendix}

\usepackage{lineno,hyperref}
\usepackage{graphicx}
\usepackage{subcaption}
\usepackage{tikz}
\usetikzlibrary{arrows,calc}
\usetikzlibrary{patterns}
\tikzset{
	>=stealth',
	help lines/.style={dashed, thick},
	axis/.style={<->},
	important line/.style={thick},
	connection/.style={dotted},
}

\usepackage{authblk}

\usepackage{algorithm}
\usepackage{amsthm}
\usepackage{algpseudocode}
\usepackage{amsmath}

\newtheorem{theorem}{Theorem}
\newtheorem{lemma}{Lemma}

\newtheorem{corollaty}{Corollary}

\def\x{\mathbf{x}}

\def\0{\mathbf{0}}

\def\cC{L'}

\newcommand{\taba}{\hspace*{0.25in}}
\newcommand{\tabb}{\hspace*{.5in}}

\title{Network Flow Methods for the Minimum Covariates Imbalance Problem}
\author[1]{Dorit S. Hochbaum\thanks{Research supported in part by NSF award No. CMMI-1760102.}
\thanks{hochbaum@ieor.berkeley.edu}
}
\author[1]{Xu Rao\thanks{xrao@berkeley.edu}}
\affil[1]{\small Department of Industrial Engineering and Operations Research, \ University of California, Berkeley}
\date{}

\begin{document}
	\maketitle

%
%
%
%
%
%
%
%
	
\begin{abstract}
	The problem of balancing covariates arises in observational studies where one is given a group of control samples and another group, disjoint from the control group, of treatment samples. Each sample, in either group, has several observed nominal covariates. The values, or categories, of each covariate partition the treatment and control samples to a number of subsets referred to as \textit{levels} where the samples at every level share the same covariate value. We address here a problem of selecting a subset of the control group so as to balance, to the best extent possible, the sizes of the levels between the treatment group and the selected subset of control group, the {\em min-imbalance} problem.
	It is proved here that the min-imbalance problem, on two covariates, is solved efficiently with network flow techniques.  We present an integer programming formulation of the problem where the constraint matrix is totally unimodular, implying that the linear programming relaxation to the problem has all basic solutions, and in particular the optimal solution, integral.   This integer programming formulation is linked to a minimum cost network flow problem which is solvable in $O(n\cdot (n' + n\log n))$ steps, for $n$ the size of the treatment group and $n'$ the size of the control group.  A more efficient algorithm is further devised based on an alternative, maximum flow, formulation of the two-covariate min-imbalance problem, that runs in $O(n'^{3/2}\log^2n)$ steps.
\end{abstract}

%
\section{Introduction}

The problem of balancing covariates arises in observational studies where one is given a group of control samples and another group, disjoint from the control group, of treatment samples. Each sample, in either group, has several observed nominal covariates. The values, or categories, of each covariate partition the treatment and control samples to a number of subsets referred to as \textit{levels} where the samples at every level share the same covariate value.

When estimating causal effects using observational data, it is desirable to replicate a randomized experiment as closely as possible by obtaining treatment and control groups with similar covariate distributions. This goal can often be achieved by choosing well-matched samples of the original treatment and control groups, thereby reducing bias due to the covariates. The matching is typically one-to-one.

The bias due to the covariates can be completely removed if for each matched pair, the two samples belong to the same levels over all covariates.  Satisfying this requirement, however, typically results in a very small selection from the treatment and control group, which is not desirable. To address this Rosenbaum et al. \cite{Rosenbaum07} introduced a relaxation requiring instead to match all treatment samples to a subset of the control samples so that the number of control and treatment samples in each level of each covariate are the same.  This requirement is known in the literature as {\em fine balance}.

It is not always feasible to find a selection of the control samples that satisfies the fine balance requirement.  To that end several papers considered the goal of minimizing the violation of this requirement, which we refer to as {\em imbalance}. Yang et al. \cite{Yang12} considered finding an optimal matching between the treatment samples and a selection of the control samples where that selection minimizes an imbalance objective, to be explained below.  Bennett et al. \cite{BVZ20} considered directly the minimum imbalance problem that finds a selection which minimizes the violation of the fine balance requirement.

Our focus here is the minimum imbalance problem, also referred to as the {\em min-imbalance} problem.
The min-imbalance problem is trivial for a single covariate and NP-hard for three or more covariates, as discussed next. For the case of two covariates, we introduce here efficient network flow algorithms that solve the problem in polynomial time.

Let $P$ be the number of covariates to be balanced. To introduce the min-imbalance problem, consider first the case of a single covariate, $P=1$, that partitions the control and treatment groups into, say, $k$ levels each. Let the sizes of the levels in the treatment group be $\ell_1, ..., \ell_k$ and the sizes of the levels in the control group be $\ell'_1, ..., \ell'_k$. The min-imbalance problem is to select a subset of the control group, called {\em selection}, such that the sizes of the levels in the selection, approximate as closely as possible the sizes $\ell_1, ..., \ell_k$. The solution to the single covariate min-imbalance problem is trivial:  the optimal selection is any $\ell_i$ samples out of the $\ell'_i$ level $i$ control samples if $\ell'_i\geq \ell_i$; otherwise, select all $\ell'_i$ level $i$ control samples. The value of the objective function corresponding to this selection is then $\sum_{i=1}^k \max\{0,\ell_i-\ell'_i\} $, which is the optimal value for the single covariate min-imbalance problem.

For the case of multiple covariates, covariate $p$ partitions both treatment and control groups into $k_p$ levels each. Denote the sizes of the levels in the treatment group under covariate $p$ by $\ell_{p,1}, \ell_{p,2}, ..., \ell_{p,k_p}$, and let the partition of the control group under covariate $p$ be  $L'_{p,1}, L'_{p,2}, ..., L'_{p,k_p}$ of sizes $\ell'_{p,1}, \ell'_{p,2}, ..., \ell'_{p,k_p}$. The min-imbalance problem is to find a selection $S$, subset of the control group, that minimizes the objective function: $\sum_{p=1}^P \sum_{i=1}^{k_p} | |S\cap L'_{p,i}|-\ell_{p,i} |$.

Let the number of treatment samples be $n$ and the number of control samples be $n'$.  In applications of the min-imbalance problem it is often the case that the size of selection $S$ is required to be $n$, the size of the treatment group. This is assuming that $n\leq n'$.   Here we focus the presentation on the case where the size of $S$ is required to be $n$.  Yet, all our methods generalize to the min-imbalance problem for any specified size $q$, $q \leq n'$, of the selection, as shown in Appendix \ref{sec:sizeq}.

The min-imbalance problem was recently studied by Bennett et al. \cite{BVZ20} where a mixed integer programming (MIP) formulation of the problem was presented, with the size of the selection equal to $n$, the size of the treatment group. They showed that the corresponding linear programming (LP) relaxation yields an integer solution when the number of covariates $P\leq 2$.  Consequently, Bennett et al. \cite{BVZ20} established that the min-imbalance problem on two covariates is solved in polynomial time via linear programming. For $P=3$ they presented an example where the LP relaxation's solution is not integral.  It is further noted in \cite{BVZ20}, that the min-imbalance problem for $P\geq 3$ is NP-hard.  However, no proof or reference is provided.  Among other contributions, we fill this gap  with a formal proof that the min-imbalance problem for three or more covariates is NP-hard.

A problem related to the min-imbalance that selects a subset of the control samples so as to achieve a certain type of balancing of the levels is called {\em Optimal Matching with Minimal Deviation from Fine Balance}, see Yang et al. \cite{Yang12}.  We refer to this problem here for short as {\em $\kappa$-Matching-Balance} (MB). The MB problem is to optimize the assignment of each treatment sample to a number, $\kappa$, of control samples, based on a distance measure between each treatment and each control sample.  This is subject to a constraint that the selection of the $\kappa n$ control samples is an optimal solution to a problem related to min-imbalance.  In other words, MB is defined in two stages:
In the first stage the goal is to find a selection $S$ of control samples of size $\kappa n$, where $\kappa$ is an integer between $1$ and $\frac{n'}{n}$, and $n$ is the size of the treatment group, so as to minimize the $\kappa$-imbalance, defined as $\sum_{p=1}^P \sum_{i=1}^{k_p} | {|S\cap L'_{p,i}|}-\kappa \ell_{p,i} |$.  In the second stage, among all selections of control samples that attain the minimum $\kappa$-imbalance, one chooses the selection that minimizes the distance matching of each treatment sample to exactly $\kappa$ selected control samples.
Yang et al. \cite{Yang12} studied this MB problem for a single covariate and proposed two network flow algorithms.  We observe here that, for {\em any} number of covariates, if the first stage problem of minimizing the $\kappa$-imbalance has a unique solution, in terms of the number of selected samples from each level-intersection (to be discussed in the next section, Section \ref{sec:prelim}) of the control group, then an optimal solution to the second stage, and therefore to the MB problem, is easily attained by solving a network flow problem.   As shown in Appendix \ref{sec:kappa} the minimum $\kappa$-imbalance is equivalent to the min-imbalance problem, and therefore an optimal solution to a corresponding min-imbalance problem provides an optimal solution to the minimum $\kappa$-imbalance problem.  It follows in particular that for a single covariate, the first stage problem of minimizing the $\kappa$-imbalance is trivial and has a unique solution in terms of the number of selected samples from each level of the control group.  For two or more covariates, no polynomial running time algorithm is known for this problem.  However, for two covariates, any of the algorithms shown here, also solve optimally the minimum $\kappa$-imbalance problem.
For any number of covariates, when the minimum $\kappa$-imbalance solution is unique (in terms of level-intersection sizes) and known, MB problem is solved in polynomial time in the second stage (for details see Appendix \ref{sec:kappa}). In particular for two covariates, if the solution to the minimum $\kappa$-imbalance problem is unique, then the problem is solved efficiently with network flow techniques.

Our main results here are efficient algorithms for the min-imbalance problem with two covariates.  We present an integer programming formulation of the problem, related to that of Bennett et al \cite{BVZ20}, and show that the constraint matrix is totally unimodular. That implies the linear programming relaxation to the problem has all basic solutions, and in particular the optimal solution, integral.   We then show that the two-covariate min-imbalance problem can be solved much more efficiently than as a linear program with network flow techniques.  We show how to solve the problem as a minimum cost network flow problem with complexity of $O(n\cdot (n' + n\log n))$.  We then devise a more efficient algorithm based on a maximum flow formulation of the two-covariate min-imbalance problem that runs in $O(n'^{3/2}\log^2n)$ steps.
Another contribution here is a proof that for three or more covariates, the min-imbalance problem is NP-hard.

In summary our contributions here are:
\begin{itemize}	
	\item An integer programming formulation of the min-imbalance problem with two covariates that has a totally unimodular constraint matrix.
	\item A formulation of the two-covariate min-imbalance problem as a minimum cost network flow problem, solved in $O(n\cdot (n' + n\log n))$ steps.
	\item A maximum flow formulation of the two-covariate min-imbalance problem that is solved in $O(n'^{3/2}\log^2n)$ steps.
	\item An NP-hardness proof of the min-imbalance problem with three or more covariates.
\end{itemize}

{\bf Paper overview:}  Section  \ref{sec:prelim} provides the notation used here.
Section  \ref{sec:ip} describes the integer programming formulation with a totally unimodular constraint matrix for the two-covariate min-imbalance problem.   We describe the minimum cost network flow formulation of the two-covariate min-imbalance problem, and the algorithm used to solved it in Section \ref{sec:MCNF}.  In Section \ref{sec:max-flow} we provide a maximum flow problem, the output of which is converted, in linear time, to an optimal solution to the two-covariate min-imbalance problem.  We provide several concluding remarks in Section \ref{sec:conclusions}. The NP-hardness proof of the min-imbalance problem with three of more covariates is provided in Appendix \ref{sec:NPC}. Appendix \ref{sec:sizeq} explains our method also applies to the min-imbalance problem for other specified size of the selection.

\section{Preliminaries and notations} \label{sec:prelim}
The goal of the min-imbalance problem is to identify a selection of control samples that will match, at each level of each covariate partition, as closely as possible, the number of treatment samples in that level and the number of selection samples in the same level.  We denote by $P$ the number of covariates and the number of levels in the partition induced by covariate $p$ by $k_p$, for $p=1, ..., P$.
Let the levels of the treatment group under covariate $p=1,\ldots ,P$ be $L_{p,1}, L_{p,2}, ..., L_{p,k_p}$ of sizes $\ell_{p,1}, \ell_{p,2}, ..., \ell_{p,k_p}$, and let the levels of the control group under covariate $p$ be  $L'_{p,1}, L'_{p,2}, ..., L'_{p,k_p}$ of sizes $\ell'_{p,1}, \ell'_{p,2}, ..., \ell'_{p,k_p}$.  For a selection $S$ of control samples we define the discrepancy at level $i$ under covariate $p$ as,
$$dis(S,p,i)= |S\cap L'_{p,i}|-\ell_{p,i}.$$
The discrepancy of a level can be positive or negative. If the discrepancy is positive we refer to it as excess which is defined as $e_{p,i}(S)= \max \{ 0, dis(S,p,i) \}$, and if negative, we refer to it as deficit $d_{p,i}(S)= \max \{ 0, -dis(S,p,i) \}$.  With this notation the imbalance of a selection $S$, $IM(S)$ is,
 $$IM(S)=\sum_{p=1}^P\sum_{i=1}^{k_p} \left(e_{p,i}(S)+d_{p,i}(S)\right) .$$
 Notice that this quantity is identical to the imbalance form presented in the introduction: $IM(S)=\sum_{p=1}^P \sum_{i=1}^{k_p} | |S\cap L'_{p,i}|-\ell_{p,i} |$.

We now present the integer programming formulation that was given by Bennett et al.\ in \cite{BVZ20} for the min-imbalance problem.  That integer program involves two sets of decision variables: 
\\ $z_j$: a binary variable equal to $1$ if control sample $j$ is in the selection $S$, and $0$ otherwise, for $j=1,\ldots,n'$;
\\ $y_{p,i}=|dis(S,p,i)|=| |S\cap L'_{p,i}|-\ell_{p,i}|$: the absolute value of discrepancy at level $i$ under covariate $p$, for $p=1,\ldots,P$, and $i=1,\ldots,k_p$.  We note that this variable can only assume integer values, even though Bennett et al.\ refer to this as a {\em mixed} integer programming formulation.\\
With these variables the formulation is:

\begin{subequations}\label{MI0}
	\begin{align}
	& \min& \sum_{p=1}^P\sum_{i=1}^{k_p} y_{p,i} &&\\
	&\text{s.t.} &  \sum_{j\in \cC_{p,i}} z_j - \ell_{p,i} \leq y_{p,i}  &&   p=1,\ldots,P, \quad  i=1,\ldots,k_p
	\label{cons:ab1} \\
	&&   \ell_{p,i} - \sum_{j\in \cC_{p,i}} z_j \leq y_{p,i}  &&  p=1,\ldots,P, \quad  i=1,\ldots,k_p
	 \label{cons:ab2} \\
	&& \sum_{j=1}^{n'} z_{j} = n && \label{cons:sum}  \\
	&& z_j\in\{0, 1\} && j=1,\ldots,n'.
	\end{align}
\end{subequations}

For each pair $p,i$, $p=1,\ldots,P$ and $i=1,\ldots,k_p$,
constraint (\ref{cons:ab1}) and (\ref{cons:ab2}) ensure that $y_{p,i}$ assumes the absolute value of the difference between the number of selected level $i$ control samples and $\ell_{p,i}$ at an optimal solution.  These constraints also ensure that any feasible $y_{p,i}$ is non-negative and therefore a non-negativity constraint is not required for variable $y_{p,i}$.  The constraint (\ref{cons:sum}) specifies that the size of selected subset equals the size of the treatment group.

Bennett et al.\ \cite{BVZ20} proved that any basic solution of the linear programming relaxation of (\ref{MI0}) is integral for $P=2$.   We provide in Section \ref{sec:ip} a stronger result showing that a slightly modified form of this integer programming formulation has a constraint matrix which is totally unimodular for $P=2$.  That implies that result that every basic solution is integer, and furthermore the constraint matrix is that of a minimum cost network flow problem.  This implies that the problem can be solved significantly more efficiently than as linear programming problem.

An optimal solution to formulation (\ref{MI0}) specifies for each control sample whether or not it is in the selection. We observe however that the output to the min-imbalance problem, for any number of covariates, can be presented more compactly in terms of {\em level-intersections}.  For $P$ covariates the intersection of the level sets $L'_{1,i_1} \cap  L'_{2,i_2}\cap \ldots \cap L'_{P,i_P}$, $i_p=1, \ldots ,k_p$, $p=1, \ldots ,P$, form a partition of the control group.   The number of non-empty sets in this partition is at most $\min \{ n' , \Pi _{p=1}^P k_p \}$.  Instead of specifying which sample belongs to the selection, it is sufficient to determine the number of selected control samples in each level intersection, since the identity of the specific selected samples has no effect on the imbalance. With this discussion, we have a theorem on the presentation of the solution to the min-imbalance problem in terms of the level-intersection sizes.

\begin{theorem}\label{thm:unique}
	The level-intersection sizes $s_{i_1,i_2,\ldots ,i_P}$ are an optimal solution to the min-imbalance problem if there exists an optimal selection $S$ such that $s_{i_1,i_2,\ldots ,i_P}=|S \cap L'_{1,i_1} \cap  L'_{2,i_2}\cap \ldots \cap L'_{P,i_P}|$.
\end{theorem}
  We will say that the solution to the problem is unique if for any optimal selection $S$, the numbers $s_{i_1,i_2,\ldots ,i_P}=|S \cap L'_{1,i_1} \cap  L'_{2,i_2}\cap \ldots \cap L'_{P,i_P}|$ are unique.
  In order to derive an optimal selection given the optimal level-intersection sizes, one selects any $s_{i_1,i_2,\ldots ,i_P}$ control samples from the intersection $L'_{1,i_1} \cap  L'_{2,i_2}\cap \ldots \cap L'_{P,i_P}$ for $i_p=1, \ldots ,k_p$, $p=1, \ldots ,P$. Obviously, when $s_{i_1,i_2,\ldots ,i_P}<|L'_{1,i_1} \cap  L'_{2,i_2}\cap \ldots \cap L'_{P,i_P}|$ for at least one combination of $i_1,i_2,\ldots ,i_P$, the optimal selection is never unique.

Standard formulations of the minimum cost network flow (MCNF) and the maximum flow problems as well as a discussion of why MCNF constraint matrix is totally unimodular, are given in Appendix \ref{sec:flow}.

\section{A modified formulation with a totally unimodular constraint matrix for $P=2$}\label{sec:ip}


In our formulation, instead of using variables $y_{p,i}$, we use variables for excess and deficit. As discussed in Section \ref{sec:prelim}, $||S\cap L'_{p,i}| - \ell_{p,i}| =e_{p,i}(S)+d_{p,i}(S)$ for each $p$ and $i$.   We let the variable for excess for $p$ and $i$ be $e_{p,i}$ and the variable for deficit be $d_{p,i}$.  Note that $y_{p,i}=e_{p,i}+d_{p,i}$ where both $e_{p,i}$ and $d_{p,i}$ are non-negative variables.


For each $p$ and $i$, $|S\cap L'_{p,i}|  - \ell_{p,i} = e_{p,i}-d_{p,i} \Leftrightarrow |S\cap L'_{p,i}|+d_{p,i} -e_{p,i} =  \ell_{p,i} $.

In the modified formulation shown below, the constraints (\ref{z-cons:de1}) and (\ref{z-cons:de2}) for the two covariates are separated to facilitate the identification of the total unimodularity property. Since $L'_{1,1} ,..., L'_{1,k_1}$ is a partition of the control group, $\sum_{i=1}^{k_1}  |S\cap L'_{1,i}|  = |S|$. Also, because $\ell_{1,1}, ..., \ell_{1,k_1}$ are the sizes of the levels partition of the treatment group for the first covariate, it follows that $\sum_{i=1}^{k_1} \ell_{1,i}  = n$. Therefore, $\sum_{i=1}^{k_1} \left(e_{1,i}-d_{1,i}\right)  = \sum_{i=1}^{k_1} \left(|S\cap L'_{1,i}|  - \ell_{1,i} \right) = |S| - n$ . So specifying $|S| = n$ is equivalent to constraint (\ref{z-cons:sum1}) in formulation (\ref{MI_multi}) given below:


%
%
\begin{subequations}\label{MI_multi}
	\begin{align}
	& \min & \sum_{p=1}^2\sum_{i=1}^{k_p} \left(e_{p,i}+d_{p,i}\right) &\label{z-obj} \\
	&\text{s.t.} &  \sum_{j\in \cC_{1,i}} z_j +d_{1,i} -e_{1,i} =  \ell_{1,i}  &&   i = 1,...,k_1 \label{z-cons:de1} \\
	&&  \sum_{j\in \cC_{2,i}} z_j +d_{2,i} - e_{2,i} =  \ell_{2,i} &&i = 1,...,k_2  \label{z-cons:de2}\\
	&& - \sum_{i=1}^{k_1} d_{1,i} + \sum_{i=1}^{k_1} e_{1,i} = 0 && \label{z-cons:sum1} \\
	&& e_{p,i}, d_{p,i} \geq 0 && p=1,2, \quad i = 1,...,k_p   \label{z-cons:nonneg} \\
	&& z_j\in\{0, 1\} && j=1,...,n' \label{z-cons:bin}.
	\end{align}
\end{subequations}

\begin{lemma}\label{lem:TU}
	The constraint matrix of LP relaxation of formulation (\ref{MI_multi}) is totally unimodular.
\end{lemma}
\begin{proof}
First, it is noted that in the constraint matrix of (\ref{MI_multi}) each entry is $0$, $1$ or $-1$.
Multiplying some rows of a totally unimodular matrix by $-1$ preserves total unimodularity. Consider the matrix resulting by multiplying the rows of constraint (\ref{z-cons:de2}) by $-1$. As shown next, each column in this new matrix has at most one $1$ and at most one $-1$. Hence, by a well-known theorem (see Appendix \ref{sec:flow} Theorem \ref{thm:TU}) the matrix is totally unimodular.
	 \begin{itemize}
	 	 \item[(1)] Both $\{\cC_{1,1}, ..., \cC_{1,k_1}\}$ and $\{\cC_{2,1}, ..., \cC_{2,k_2}\}$ are partitions of the control group, so $\cC_{1,1}, ..., \cC_{1,k_1}$ are mutually disjoint, and $\cC_{2,1}, ..., \cC_{2,k_2}$ are mutually disjoint. The column of each $z_j$ has exactly one $1$ in rows corresponding to (\ref{z-cons:de1}), and one $-1$ (after multiplication) in rows corresponding to (\ref{z-cons:de2}).
	 	 \item[(2)] For each $i$, the column of $d_{1,i}$ has exactly one $1$ in rows corresponding to (\ref{z-cons:de1}) and exactly one $-1$ in rows corresponding to (\ref{z-cons:sum1}); the column of $e_{1,i}$ has exactly one $-1$ in rows corresponding to (\ref{z-cons:de1})) and exactly one $1$ in rows corresponding to (\ref{z-cons:sum1}).
	 	 \item[(3)] For each $i$, the column of $d_{2,i}$ has exactly one non-zero, $1$ or $-1$, entry in rows corresponding to (\ref{z-cons:de2}); the column of $e_{2,i}$ has exactly one non-zero, $1$ or $-1$, entry in rows corresponding to (\ref{z-cons:de2}).
	 \end{itemize}

	Since the new matrix has at most one $1$ and at most one $-1$ in each column, it is totally unimodular.
\end{proof}

Formulation (\ref{MI_multi}) is in fact also a minimum cost network flow (MCNF) formulation. See Appendix \ref{sec:flow} for a generic formulation of MCNF. A generic MCNF formulation has exactly one $1$ and one $-1$ in each column of the constraint matrix.  To make formulation (\ref{MI_multi}) have this structure, we multiple all coefficients in constraint (\ref{z-cons:de2}) by $-1$, and add a redundant constraint $\sum_{i=1}^{k_2} d_{2,i} - \sum_{i=1}^{k_2} e_{2,i} = 0$, which, like constraint (\ref{z-cons:sum1}), is also equivalent to $|S| = n$.
In the next section, we streamline this formulation to derive a more efficient network flow formulation for this problem.

\section{Network flow formulation for $P=2$} \label{sec:MCNF}

Here we use the level-intersection sizes as variables, $x_{i_1,i_2}$ for $i_1=1,...,k_1, i_2=1,...,k_2$. These variables can also be written as $x_{i_1,i_2}= \sum _{j \in \cC_{1,i_1} \cap \cC_{2,i_2}} z_j$.
With these decision variables we get the following network flow formulation:  
\begin{subequations}\label{MI}
	\begin{align}
	& \min & \sum_{p=1}^2\sum_{i=1}^{k_p} \left(e_{p,i}+d_{p,i}\right) &\label{obj} \\
	&\text{s.t.} &  \sum_{i_2=1}^{k_2}  x_{i_1,i_2} +d_{1,i_1} -e_{1,i_1} = \ell_{1,i_1}   && i_1=1,...,k_1  \label{cons:de1} \\
	&&  -\sum_{i_1=1}^{k_1}  x_{i_1,i_2} -d_{2,i_2} +e_{2,i_2} = -\ell_{2,i_2}   && i_2=1,..., k_2 \label{cons:de2}\\
	&& - \sum_{i_1=1}^{k_1} d_{1,i_1} + \sum_{i_1=1}^{k_1} e_{1,i_1} = 0 && \label{cons:sum1} \\
	&&  \sum_{i_2=1}^{k_2} d_{2,i_2} - \sum_{i_2=1}^{k_2} e_{2,i_2} = 0 && \label{cons:sum2}\\
	&& e_{p,i}, d_{p,i} \geq 0 &&  p=1,2, \quad k = 1,...,k_p  \label{cons:nonneg} \\
	&& 0\leq x_{i_1,i_2} \leq u_{i_1,i_2}&&  i_1=1,...,k_1, \quad i_2=1,..., k_2 \label{cons:bin}
	\end{align}.
\end{subequations}

The formulation (\ref{MI}) is a minimum cost network flow problem.  The corresponding network is shown in Figure \ref{fig:mcnf}, where all capacity lower bounds are $0$, and each arc has a cost per unit flow and upper bound associated with it.  Nodes of type $(1,i_1)$  each has supply of $\ell _{1,i_1}$ and nodes of type $(2,i_2)$  each has demand of $\ell _{2,i_2}$ In Figure \ref{fig:mcnf}, for each $i_1$ and $i_2$, the flow on the arc between node $(1,i_1)$ and node $(2,i_2)$ represents variable $x_{i_1,i_2}$; arc from node $1$ to node $(1,i_1)$ represents the excess $e_{1,i_1}$; arc to node $1$ from any node $(1,i_1)$ represents the deficit $d_{1,i_1}$; arc from node $2$ to any node $(2,i_2)$ represents the deficit $d_{2,i_2}$; arc to node $2$ from any node $(2,i_2)$ represents the excess $e_{2,i_2}$. So the per unit arc cost should be 1 for arcs between node 1 or 2 and any node in $\{(1,1), (1,2), ..., (1,k_1) \}\cup \{(2,1), (2,2), ...,(2,k_2)\}$. All other arcs have cost $0$. It is easy to verify that constraints (\ref{cons:de1}) are corresponding to the flow balance at nodes $(1,i_1)$ for all $i_1$, constraints (\ref{cons:de2}) are corresponding to the flow balance at nodes $(2,i_2)$ for all $i_2$.    Constraint (\ref{cons:sum1}) is corresponding to the flow balance at node 1, and constraint (\ref{cons:sum2}) is corresponding to the flow balance at node 2.

\begin{figure}[htb]
	\centering
	\begin{tikzpicture}	
	\begin{scope}[every node/.style={circle,thick,draw}]
	\node (11) at (0,0) {1,1};
	\node (12) at (0,-3) {1,2};
	\node (1k1) at (0,-9) {1,$k_1$};
	\node (21) at (5,0) {2,1};
	\node (22) at (5,-3) {2,2};
	\node (2k2) at (5,-9) {2,$k_2$};
	\node (1) at (-3,-4.5) {$1$};
	\node (2) at (8,-4.5) {$2$};
	\end{scope}

	\begin{scope}[every node/.style={thick}]
	\node (1.) at (0,-6) {$\vdots$};
	\node (2.) at (5,-6) {$\vdots$};
	\node (legend) at (2.5,2) {arc legend: (cost, upperbound)};
	
	\node at (0,0.8) {$(\ell_{1,1})$};
	\node at (0,-2.2) {$(\ell_{1,2})$};
	\node at (0,-9.8) {$(\ell_{1,k_1})$};
	\node at (5,0.8) {$(-\ell_{2,1})$};
	\node at (5,-2.2) {$(-\ell_{2,2})$};
	\node at (5,-9.8) {$(-\ell_{2,k_2})$};
	
	\end{scope}
	
	\begin{scope}[
	every edge/.style={draw=black, thick,font=\small,align=center, auto}]
	\path [->] (11) edge node {$(0, u${\tiny $_{1,1}$}$)$} (21);
	\path [->] (11) edge node[right] {$(0, u${\tiny $_{1,2}$}$)$} (22);
	\path [->] (11) edge node[right] {$(0, u${\tiny $_{1,k_2}$}$)$} (2k2);
	\path [->] (12) edge node {$(0, u${\tiny $_{2,2}$}$)$}  (22);
	\path [->] (1k1) edge node {$(0, u${\tiny $_{k_1,k_2}$}$)$}  (2k2);
	
	\draw [->] (1) edge [out=80,in=200] node[above,left] {$(1,\infty)$} (11);
	\draw [->] (1) edge [out=50,in=185] node[above,right] {$(1,\infty)$} (12);
	\draw [->] (1) edge [out=-45,in=110] node[above,left] {$(1,\infty)$} (1k1);
	\draw [->] (11) edge [out=250,in=60] node[below,left] {$(1,\infty)$} (1);
	\draw [->] (12) edge [out=215,in=20] node[below] {$(1,\infty)$}  (1);
	\draw [->] (1k1) edge [out=160,in=-80] node[below,left] {$(1,\infty)$} (1);
	
	\draw [->] (21) edge [in=100,out=-20] node[above,right] {$(1,\infty)$} (2);
	\draw [->] (22) edge [in=130,out=-5] node[above,left] {$(1,\infty)$} (2);
	\draw [->] (2k2) edge [in=225,out=70] node[above,right] {$(1,\infty)$} (2);
	\draw [->] (2) edge [in=-70,out=120] node[below,right] {$(1,\infty)$} (21);
	\draw [->] (2) edge [in=-35,out=160] node[below] {$(1,\infty)$}  (22);
	\draw [->] (2) edge [in=20,out=-100] node[below,right] {$(1,\infty)$} (2k2);

	\end{scope}
	
	\end{tikzpicture}
	\caption{MCNF graph corresponding to formulation (\ref{MI}).}
	\label{fig:mcnf}
\end{figure}
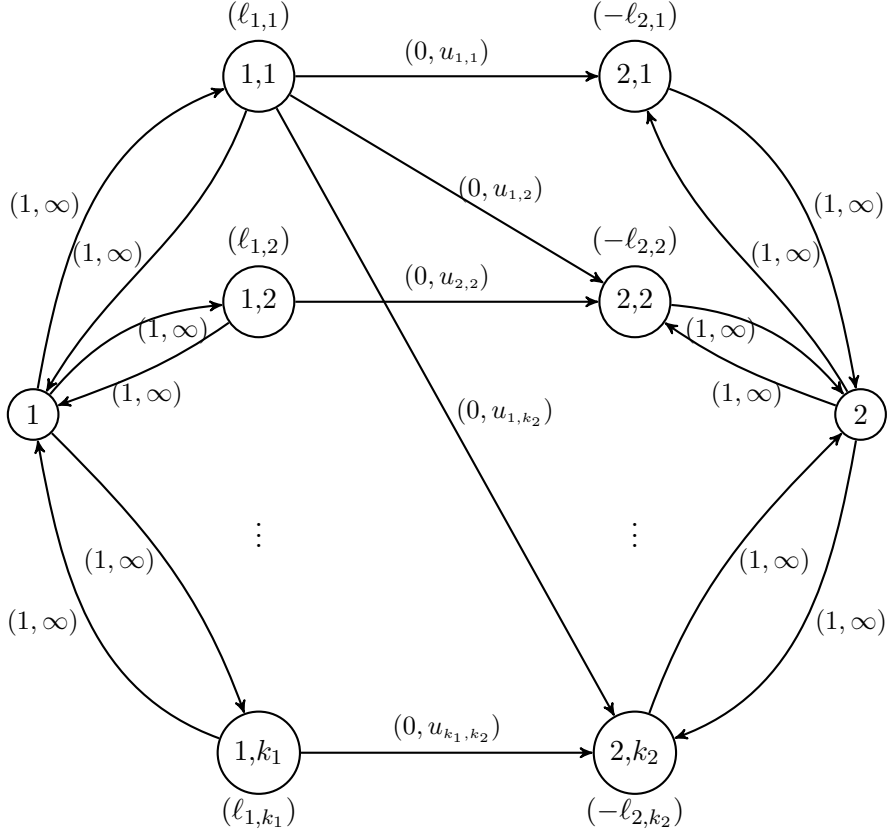

\begin{theorem}\label{thm:timeMCNF}
	The 2-covariate min-imbalance problem is solved as a minimum cost network flow problem in $O(n\cdot (n' + n\log n))$ time.
\end{theorem}
\begin{proof}
We choose the algorithm of {\em successive shortest paths} that is particularly efficient for a MCNF with ``small" total supply to solve the network flow problem of the 2-covariate min-imbalance problem.

The successive shortest path algorithm iteratively selects a node $s$ with excess supply (supply not yet sent to some demand node) and a node $t$ with unfulfilled demand and sends flow from $s$ to $t$ along a shortest path in the residual network \cite{Jewell58}, \cite{Iri60}, \cite{BG61}. The algorithm terminates when the flow satisfies all the flow balance constraints. Since at each iteration, the number of remaining units of supply to be sent is reduced by at least one unit, the number of iterations is bounded by the total amount of supply.  For our problem the total supply is $O(n)$.

At each iteration, the shortest path can be solved with Dijkstra's algorithm of complexity $O(|A| + |V| \log |V|)$, where $|V|$ is number nodes and $|A|$ is number of arcs \cite{Tomizawa71}, \cite{EK72}. In our formulation, $|V|$ is $O\left(k_1+k_2\right)$, which is at most $O(n)$. Since the number of nonempty sets $\cC_{1,i_1} \cap \cC_{2,i_2}$ is at most $\min\{n', k_1k_2 \}$, the number of arcs $|A|$ is $O(\min\{n', k_1k_2 \})$.


Hence, the total running time of applying the successive shortest path algorithm with node potentials on our formulation is $O(n\cdot (n' + n\log n))$.
\end{proof}

\section{Maximum flow formulation for $P=2$}\label{sec:max-flow}
Here we show a maximum flow (max-flow) formulation (see Appendix \ref{sec:flow} for a generic formulation of max-flow problem) for the min-imbalance problem with 2 covariates. Unlike the previous formulations, the maximum flow solution requires further manipulation in order to derive an optimal solution to the min-imbalance problem with 2 covariate.  That max-flow graph is illustrated in Figure (\ref{fig:maxf}). The source node $s$ can send at most $\ell_{1,i_1}$ units of flow to node $(1,i_1)$ for each $i_1=1,..., k_1$, the sink node can get at most $\ell_{2,i_2}$ units of flow from node $(2,i_2)$ for each $i_2=1,..., k_2$, and there can be a flow from node $(1,i_1)$ to node $(2,i_2)$ with amount bounded by $u_{i_1,i_2}$, for $i_1=1,..., k_1, i_2=1,..., k_2$.

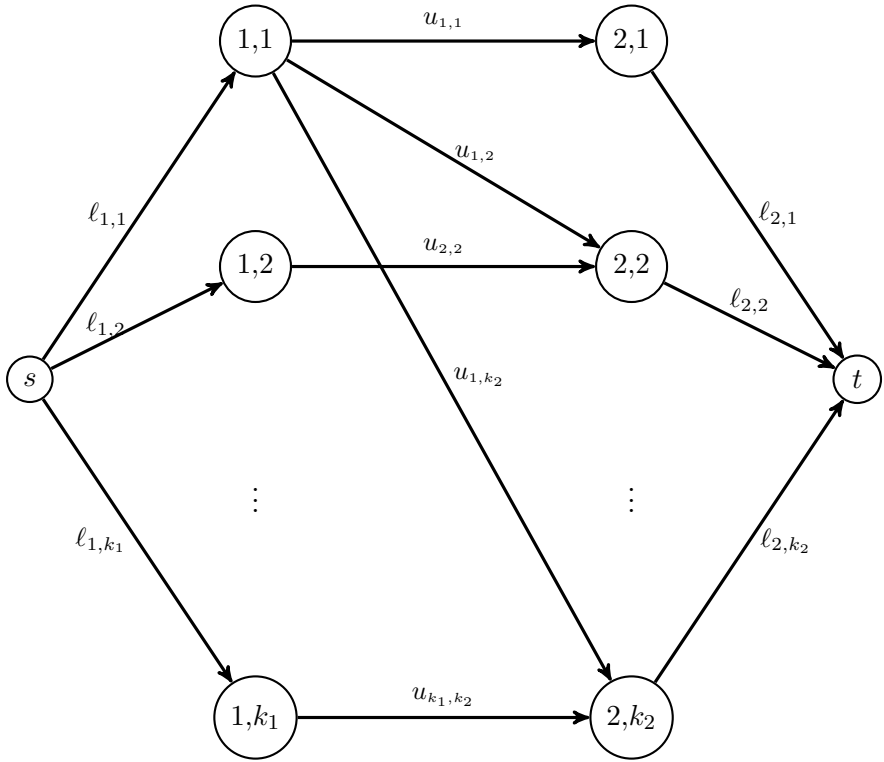
\begin{figure}[htb]
	\centering
	\begin{tikzpicture}	
	\begin{scope}[every node/.style={circle,thick,draw}]
	\node (11) at (0,0) {1,1};
	\node (12) at (0,-3) {1,2};
	\node (1k1) at (0,-9) {1,$k_1$};
	\node (21) at (5,0) {2,1};
	\node (22) at (5,-3) {2,2};
	\node (2k2) at (5,-9) {2,$k_2$};
	\node (s) at (-3,-4.5) {$s$};
	\node (t) at (8,-4.5) {$t$};
	\end{scope}

	\begin{scope}[every node/.style={thick}]
	\node (1.) at (0,-6) {$\vdots$};
	\node (2.) at (5,-6) {$\vdots$};
	\node (legend) at (2.5,2) {arc legend:  upperbound};
	
	\end{scope}
	
	\begin{scope}[
	every edge/.style={draw=black,very thick,font=\small,align=center, auto}]
	\path [->] (11) edge node {$u${\tiny $_{1,1}$}}  (21);
	\path [->] (11) edge node[right] {$u${\tiny $_{1,2}$}} (22);
	\path [->] (11) edge node[right] {$u${\tiny $_{1,k_2}$}} (2k2);
	\path [->] (12) edge node {$u${\tiny $_{2,2}$}} (22);
	\path [->] (1k1) edge node {$u${\tiny $_{k_1,k_2}$}}  (2k2);
	\path [->] (s) edge node[above,left] {$\ell_{1,1}$} (11);
	\path [->] (s) edge node[above,left] {$\ell_{1,2}$} (12);
	\path [->] (s) edge node[above,left] {$\ell_{1,k_1}$} (1k1);
	\path [->] (21) edge node[right] {$\ell_{2,1}$} (t);
	\path [->] (22) edge node[above] {$\ell_{2,2}$} (t);
	\path [->] (2k2) edge node[above, right] {$\ell_{2,k_2}$} (t);
	\end{scope}
	
	\end{tikzpicture}
	\caption{Maximum flow graph}
	\label{fig:maxf}
\end{figure}

Let the maximum flow value for the max-flow problem presented in Figure (\ref{fig:maxf}) be denoted by $f^*$, and let $\x ^*$ be the optimal flow vector, with $x^*_{i_1,i_2}$ denoting the flow amount between node $(1,i_1)$ and node $(2,i_2)$. It is obvious that $\sum_{i_1=1}^{k_1} \sum_{i_2=1}^{k_2} x^*_{i_1,i_2}=f^*\leq \sum_{i_1=1}^{k_1} \ell_{1,i_1} = n$. That means an initial selection $S'$ generated by selecting the prescribed number of control samples as in the optimal max-flow solution, i.e., selecting $x^*_{i_1,i_2}$ control samples from $\cC_{1,i_1}\cap\cC_{2,i_2}$ is of size $f^*$.  In order to get a feasible solution for the min-imbalance problem it is required to select $n-f^*$ additional control samples. The selection $S'$ has no positive excess, only non-negative deficits with respect to the levels of both covariates.  This is because $\sum_{i_2=1}^{k_2}x^*_{i_1,i_2}\leq \ell_{1,i_1}$ due to the upper bound of arc from $s$ to $(1,i_1)$ for each $i_1$, and $\sum_{i_1=1}^{k_1}x^*_{i_1,i_2}\leq \ell_{2,i_2}$ due to the upper bound of arc from $(2,i_2)$ to $t$ for each $i_2$.
To recover an optimal solution for the min-imbalance problem from the initial set $S'$, we add up to $n-f^*$ unselected control samples, one at a time, each corresponding to a level with positive deficit under either covariate $1$ or $2$. This process is repeated until either $n-f^*$ such control samples are found, or until no such control sample exists. In the latter case, to complete the size of the selection, any randomly selected control samples are added. Algorithm \ref{alg} is a formal statement of this process of recovering an optimal solution of the min-imbalance problem from the initial selection $S'$.

\begin{algorithm}[H]
	\caption{}
	\begin{algorithmic}
		\State {\bf Initialization step:} Select $x^* _{i_1,i_2}$ number of control samples from $\cC_{1,i_1}\cap\cC_{2,i_2}$ in set $S'$. 
		\While{$|S'|<n$}
		\If{there exists an unselected control sample $j\notin S'$ whose covariate 1 level is $i_1$ and covariate 2 level is $i_2$, such that $|S'\cap \cC_{1,i_1}| < \ell_{1,i_1}$ or $|S'\cap \cC_{2,i_2}| < \ell_{2,i_2}$},
		\State $S'\leftarrow S'\cup \{j\}$.
		\Else
		\State Let $S''=S'$ and let $S^+$ be any $n-|S'|$ control samples $\notin S'$.  Set $S'\leftarrow S'\cup S^+$.
		\EndIf
		\EndWhile
		\State Output $S'$.
	\end{algorithmic}\label{alg}
\end{algorithm}

To show that Algorithm \ref{alg} provides an optimal solution to the min-imbalance problem, we distinguish two cases of Algorithm \ref{alg}: (1) $S^+=\emptyset$ and (2) $|S^+|\geq 1$.
In the first case, there is, at each iteration, at least one control sample that belongs to some level with positive deficit.  In Theorem \ref{thm:maxf_easy} we prove that the output $S'$ of Algorithm \ref{alg} is an optimal solution in this case.

\begin{theorem}\label{thm:maxf_easy}
If $S^+=\emptyset$ then the output selection $S'$ of Algorithm \ref{alg} is optimal for the min-imbalance problem, with an optimal objective value of $2(n-f^*)$.
\end{theorem}

\begin{proof}
	First, we show that the total imbalance of the selection $S'$ is $IM(S')=2(n-f^*)$. At the initialization step the selection $S'$ has only deficits for all levels, with total deficit for covariate 1, $\sum_{i_1=1}^{k_1}(\ell_{1,i_1}-\sum_{i_2=1}^{k_2}x_{i_1,i_2})=n-f^*$, and total deficit for covariate 2, $\sum_{i_2=1}^{k_2}(\ell_{2,i_2}-\sum_{i_1=1}^{k_1}x_{i_1,i_2})=n-f^*$.  At each iteration, there is an added control sample, say in $\cC_{1,i_1}\cap\cC_{2,i_2}$, such that either $\cC_{1,i_1}$ or $\cC_{2,i_2}$ has a positive deficit with respect to $S'$. It is however impossible for both $\cC_{1,i_1}$ and $\cC_{2,i_2}$ to have a positive deficit with respect to $S'$ since otherwise, there is an $s,t$-augmenting path, from $s$ to node $(1,i_1)$, to node $(2, i_2)$, to $t$, along which the flow can be increased by at least one unit. This is in contradiction to the optimality of the max-flow solution $\x ^*$.  As a result, at each iteration where a control sample is added, the total deficit is reduced by one unit, and the total excess is {\em increased} by one unit. Thus, at each iteration of the {\bf if} step, the sum of total deficit and excess remains the same, namely $2(n-f^*)$.
	
Suppose, by contradiction, that there exists a selection $S^*$ for which the total imbalance is lower, $IM(S^*)<2(n-f^*)$.  We repeat the following iterative procedure of removing samples from $S^*$ until there is no positive excess remaining: while there is a level of either covariate with positive excess with respect to $S^*$, we remove one sample of $S^*$ that belongs to this level. Each such iteration results in the total excess reducing by at least 1 unit and the total deficit increasing by at most 1 unit, and therefore the sum of total deficit and excess does not increase. So when this iterative procedure ends, the total excess is zero and the total deficit is at most $IM(S^*)$. Let $x_{i_1,i_2}$ be the number of samples remaining in $S^*\cap\cC_{1,i_1}\cap\cC_{2,i_2}$ after this excess removing procedure. Since there is no positive excess, $\x$ is a feasible solution for the max-flow problem with the flow between node $(1,i_1)$ and node $(2,i_2)$ equal to $x_{i_1,i_2}$.  The sum of deficits associated with this remaining set is $n-\sum_{i_1=1}^{k_1}\sum_{i_2=1}^{k_2} x_{i_1,i_2}$ for covariate 1 and $n-\sum_{i_1=1}^{k_1}\sum_{i_2=1}^{k_2} x_{i_1,i_2}$ for covariate 2, for a total of $2(n-\sum_{i_1=1}^{k_1}\sum_{i_2=1}^{k_2} x_{i_1,i_2})$, which is at most $IM(S^*)$.  Therefore, the total flow value, $\sum_{i_1=1}^{k_1}\sum_{i_2=1}^{k_2} x_{i_1,i_2}    $, satisfies that it is at least $n- \frac{IM(S^*)}{2}$. Since $n-\frac{IM(S^*)}{2} > n-(n-f^*)=f^*$, it follows that the value of the feasible flow induced by the set $S^*$ is greater than the maximum flow value $f^*$, which contradicts the optimality of $f^*$.
\end{proof}

We now address the second case where $|S^+|\geq 1$ and $|S''|<n$.  In this case, the total imbalance of $S''$ is,  from the arguments in the proof of Theorem \ref{thm:maxf_easy}, $IM(S'')=2(n-f^*)$.
Each one of the $S^+$ samples selected adds 1 unit of excess to each covariate, resulting in the addition of 2 units of excess to the imbalance. Therefore, the total imbalance of the output solution is $2(n-f^*)+2|S^+|$.  We next show what the value of $|S^+|$ is, and then demonstrate that any feasible selection to the min-imbalance problem has total imbalance of at least $2(n-f^*)+2|S^+|$.  This will prove that the output of Algorithm \ref{alg}, $S'$, is an optimal solution to the min-imbalance problem.

It will be useful to consider an equivalent form of Algorithm \ref{alg}.
For each level $i$ of covariate $p$ that has $|S'\cap \cC_{p,i}| < \ell_{p,i}$, we add the largest number possible of available control samples in $\cC_{p,i}$ so long as the total does not exceed $n$.  This number is $\min \{ \ell_{p,i}-|S'\cap \cC_{p,i}| , \ell'_{p,i}-|S'\cap \cC_{p,i}| \}$.   Let $\bar{\ell}_{p,i}=\min\{\ell_{p,i}, \ell'_{p,i}\}$  then for each $p,i$ that has $|S'\cap \cC_{p,i}| < \ell_{p,i}$ we add $\bar{\ell}_{p,i}-|S'\cap \cC_{p,i}|$ previously unselected control samples to $S'$.  The outcome of this equivalent procedure is exactly the same as that of Algorithm \ref{alg}.  In the case that $|S^+|\geq 1$ there is an insufficient number of control samples to add to $S'$ after for all levels the largest number possible has been added.  Therefore, at the end of this process, the {\bf if} step returns that another unselected control sample does not exist, and the total number of control samples of $S''$, for each level $i$ of covariate $p$, is $\bar{\ell}_{p,i}$.


\begin{lemma}
	If $|S^+| \geq 1$ (and $|S''|= n- |S^+|<n$) then $|S^+|=n-(\bar{\ell}_1+\bar{\ell}_2-f^*)$ where $\bar{\ell}_1=\sum_{i_1=1}^{k_1} \bar{\ell}_{1,i_1}$ and $\bar{\ell}_2=\sum_{i_2=1}^{k_2} \bar{\ell}_{2,i_2}$.
\end{lemma}
\begin{proof}
	
	At the initialization step of Algorithm \ref{alg} $|S'|=f^*$ and the total deficit is $2(n-f^*)$. Each time a control sample is added to $S'$ in the {\bf if} step, the total deficit is decreased exactly by 1 unit. So we can derive the value of $|S''|$ when the algorithm terminates if we know the total deficit when the algorithm terminates.  Note that the total excess may change, but we only consider here the deficit part of the imbalance.

From the discussion above, the total number of control samples of $S''$, for each level $i$ of covariate $p$, is $\bar{\ell}_{p,i}$.   We denote $\bar{\ell}_1=\sum_{i_1=1}^{k_1} \bar{\ell}_{1,i_1}$, and $\bar{\ell}_2=\sum_{i_2=1}^{k_2} \bar{\ell}_{2,i_2}$.  Since the sum $\sum_{i_1=1}^{k_1} \ell_{1,i_1}=n$ and $\sum_{i_2=1}^{k_2} \ell_{2,i_2}=n$, the sum of  deficits of set $S''$ under covariate 1 is $\sum_{i_1=1}^{k_1} \ell_{1,i_1}-\bar{\ell}_{1,i_1} = n-\bar{\ell}_1$, and the sum of  deficits under covariate 2 equals $\sum_{i_2=1}^{k_2} \ell_{2,i_2}-\bar{\ell}_{2,i_2} = n-\bar{\ell}_2$. It follows that the sum of deficits of $S''$ is $2n-\bar{\ell}_1 -\bar{\ell}_2$.

Since the initial set $S'$ that has total deficit of $2(n-f^*)$ has its deficit reduced through Algorithm \ref{alg} to $2n-\bar{\ell}_1 -\bar{\ell}_2$ in the set $S''$, the additional number of control sample in $S''$ that were added to the initial $f^*$ control samples is $2(n-f^*)-(2n-\bar{\ell}_1-\bar{\ell}_2) = \bar{\ell}_1+\bar{\ell}_2-2f^*$. Therefore, the size of $S''$ is $f^* + (\bar{\ell}_1 + \bar{\ell}_2 - 2f^*)= \bar{\ell}_1 + \bar{\ell}_2 - f^*$.   This number is less than $n$ and the size of $S^+$ then satisfies, $|S^+|=n-(\bar{\ell}_1+\bar{\ell}_2-f^*)$.

%
	
\end{proof}

\begin{corollaty}\label{cor}
If $|S^+| \geq 1$  when Algorithm \ref{alg} terminates, the total imbalance of the output solution $S'$ is $IM(S')= 4n-2\bar{\ell}_1-2\bar{\ell}_2$.
\end{corollaty}
\begin{proof}
The imbalance of $S''$, as in the proof of Theorem \ref{thm:maxf_easy}, is equal to $2(n-f^*)$.  Since each sample in $S^+$ adds two units to the imbalance, the total imbalance of the output solution $S'$ is $IM(S')= 2(n-f^*) + 2(n-(\bar{\ell}_1+\bar{\ell}_2-f^*))$, which is equal to $4n-2\bar{\ell}_1-2\bar{\ell}_2$, as stated.
\end{proof}

Next, we prove that this is the minimum total imbalance achievable.

\begin{theorem}\label{thm:maxf_hard}
	For any selection of size $n$, the total imbalance must be greater or equal to $4n-2\bar{\ell}_1-2\bar{\ell}_2$.
\end{theorem}
\begin{proof}
	
	For the optimal selection $S^*$ of size $n$, let $IM(S^*)$ be the total imbalance of $S^*$. We first classify the samples in $S^*$ into three types, $S_1$, $S_2$ and $S_3$ that forms a partition of $S^*$, using the {\sf 3-type Classification Procedure}.
	
	In the procedure we use variable $dis(p,i)$ to denote the value of $dis(S,p,i)$, discrepancy for selection $S$ and level $i$ under covariate $p$.
	With this notation the excess of the corresponding level is $e(S,p,i)=\max\{0, dis(p,i)\}$, and the deficit is $d(S,p,i)=\max\{0, -dis(p,i)\}$.
%
	
	{\sf 3-type Classification Procedure}\\
	{\sf Initialize}  \\
	\taba $S\leftarrow S^*, S_1 \leftarrow \emptyset, S_2 \leftarrow \emptyset, S_3 \leftarrow \emptyset$; \\
	\taba Let $dis(p,i) \leftarrow  |S\cap \cC_{p,i}| - \ell_{p,i} $ for $p=1,2, \ i=1,...,k_p$;\\
	{\sf $S_1$ selection:} \\
	\taba \textbf{While} there exists a sample $j$ in $S$ whose covariate 1 level is $i_1$, covariate 2 level is $i_2$, such that $dis(1,i_1) >0$ and $dis(2,i_2) >0$, pick $j$;\\
	\tabb $S_1 \leftarrow S_1 \cup \{j\}, S \leftarrow S - \{j\}, dis(1,i_1) \leftarrow dis(1,i_1)  -1, \ dis(2,i_2) \leftarrow dis(2,i_2)  -1$; \\
	\taba \textbf{end while}; \\
	{\sf $S_2$ selection:} \\
	\taba \textbf{While} there exists a sample $j$ in $S$ whose covariate 1 level is $i_1$, covariate 2 level is $i_2$, such that $dis(1,i_1) >0$ or $dis(2,i_2) >0$, pick $j$;\\
	\tabb $S_2 \leftarrow S_2 \cup \{j\}, S \leftarrow S - \{j\}, dis(1,i_1) \leftarrow dis(1,i_1)  -1, \ dis(2,i_2) \leftarrow dis(2,i_2)  -1$; \\
	\taba {\bf end while}; \\
	{\sf $S_3$ selection:} \\
	\taba $S_3\leftarrow S$;\\
	{\sf end}.
	
	The output $S_1, S_2, S_3$ of the procedure is not unique, since it depends on the order in which samples are picked. However, the statements of the theorem hold for any output of the procedure.
	Note that in the procedure, whenever a sample is picked, any $dis(p,i)$ can only go down. For that reason, once $S_1$ selection ends, there will not be another sample in $S$ for which the discrepancy values of the corresponding levels under the two covariates are both positive.
	Furthermore, once the $S_1$ and $S_2$ selections are done, $dis(p,i)\leq 0$ for each $p,i$.
	That means, $ |S_3\cap \cC_{p,i}| \leq \ell_{p,i}$ for all $p,i$.
	
	Let the sizes of the three subsets be denoted by $s_1=|S_1|, s_2=|S_2|, s_3=|S_3|$.
	
	We claim that the total imbalance of samples in $S_3$ is $IM(S^*)-2s_1$. For the $S_1$ selection part of the procedure, each samples picked in $S_1$ reduces the total excess by 2. For each sample selected in the $S_2$ selection part of the procedure,  the excess is reduced by 1, and the deficit is increased by 1. So for each sample selected in the $S_2$ selection step, the total imbalance does not change. Therefore, the total imbalance of the samples in $S_3$ is $IM(S^*)-2s_1$.
	
	On the other hand, the total imbalance of $S_3$, which equals the sum of deficits of both covariates (all excesses equal zero as $ |S_3\cap \cC_{p,i}| \leq \ell_{p,i}$), is $2(n-s_3)$. That says, $$IM(S^*)-2s_1 = 2(n-s_3).$$  Hence,
	\[  IM(S^*)=2(n-s_3)+2s_1=2(n-s_3)+2(n-s_2-s_3)=4n-2s_2-4s_3. \] Here, the second equality comes from the fact that $s_1 + s_2 + s_3 = n$.
	
	Next, we show that $s_2\leq (\bar{\ell}_1-s_3)+(\bar{\ell}_2-s_3)$.  Let the samples in $S_2$ be ordered according to the order they were picked, $j_1, j_2,..., j_{s_2}$.  We now add these samples to $S_3$, in the reverse order $j_{s_2}, ..., j_1$. When each sample $j_q$ is added to $S_3$, the deficit is reduced by exactly 1 unit.  Once all the $S_2$ samples are added to $S_3$ the deficit at each level of $S_2 \cup S_3$ is zero, or alternatively, $dis(S_2 \cup S_3,p,i)= |(S_2 \cup S_3 )\cap \cC_{p,i}| - \ell_{p,i} \geq 0$ for each $p,i$.

We now consider the total deficit of $S_3$:  By the definition of $\bar{\ell}_1$ and $\bar{\ell}_2$, the positive deficit of $S_3$ under covariate 1 is at most $\bar{\ell}_1-s_3$ and that the positive deficit of $S_3$ under covariate 2 is at most $\bar{\ell}_2-s_3$. That means the size of $S_2$ is bounded by the amount of this deficit, $s_2\leq (\bar{\ell}_1-s_3)+(\bar{\ell}_2-s_3)$.
	
	Now we have,
	\begin{eqnarray*}
		s_2\leq (\bar{\ell}_1-s_3)+(\bar{\ell}_2-s_3)  &\Leftrightarrow &s_2 + 2s_3 \leq \bar{\ell}_1+\bar{\ell}_2\\
		&\Leftrightarrow &IM(S^*)=4n-2s_2-4s_3 \geq 4n-2\bar{\ell}_1-2\bar{\ell}_2.
	\end{eqnarray*}
	We conclude that the total imbalance $IM(S^*)$ is at least $4n-2\bar{\ell}_1-2\bar{\ell}_2$. That implies that the selection output of Algorithm \ref{alg}, $S'$, which has a total imbalance of $4n-2\bar{\ell}_1-2\bar{\ell}_2$, is optimal.
	
\end{proof}

The conclusion from Corollary \ref{cor} and Theorem \ref{thm:maxf_hard}, is that for $|S^+|\geq 1 $ when Algorithm \ref{alg} terminates, the output solution $S'$ is an optimal selection to the min-imbalance problem.
Together with Theorem \ref{thm:maxf_easy}, we have that Algorithm \ref{alg} outputs an optimal selection for the min-imbalance problem using the max-flow solution to the flow problem in Figure \ref{fig:maxf} as input.

\begin{theorem}
	The maximum flow formulation of the 2-covariate min-imbalance problem is solved the in $O(n'\cdot \min\{n^{\frac{2}{3}},n'^{\frac{1}{2}}\}\cdot\log^2 n)$ time.
\end{theorem}

\begin{proof}
	We choose here the {\em binary blocking flow} algorithm of Goldberg and Rao \cite{GR98} for solving the max-flow problem shown in Figure \ref{fig:maxf} because this algorithm depends on the maximum arc capacity which is a small quantity in our formulation.
	
	The complexity of the binary blocking flow algorithm for a graph ${G}=(V,A)$ is $O(|A|\cdot \min\{|V|^{\frac{2}{3}},|A|^{\frac{1}{2}}\}\cdot\log\frac{|V|^2}{|A|}\log U )$ where $|V|$ is number of nodes, $|A|$ is number of arcs, and $U$ is maximum arc capacity.  As argued earlier for the minimum cost network flow formulation, the number of nodes in the network $|V|$ is $O\left(k_1+k_2\right)$, which is no more than $O(n)$; and the number of arcs is bounded by $\min\{n', k_1k_2 \}$.
	Although $u_{i_1,i_2}$ could be as large as $n'$, a feasible flow to our maximum flow formulation can not have more than $\ell_{1,i_1}$ units of flow on the arc from node $(1,i_1)$ to node $(2,i_2)$. Thus the maximum arc capacity $U$ is effectively $O(n)$. The ratio $\frac{|V|^2}{|A|}\leq \frac{(k_1+k_2)^2}{k_1+k_2}\leq n$. Hence, the running time of applying the binary blocking flow algorithm to our max-flow problem is $O(n'\cdot \min\{n^{\frac{2}{3}},n'^{\frac{1}{2}}\}\cdot\log^2 n)$.
	
	The complexity of Algorithm \ref{alg} is $O(n)$ as the number of iterations is bounded by $n$, and each iteration takes $O(1)$ steps.
	
	Therefore, the running time of solving the min-imbalance problem as a max-flow problem is $O(n'\cdot \min\{n^{\frac{2}{3}},n'^{\frac{1}{2}}\}\cdot\log^2 n)$.
\end{proof}



\section{Conclusions}\label{sec:conclusions}
We present here new insights to the minimum imbalance problem that involves selection of control samples so that excess and deficit of the respective treatments samples' levels are minimized.  We demonstrate that an integer programming formulation of the problem on two covariates has a totally unimodular constraint matrix.  We then present two efficient approaches to solve the problem for two covariates.  One is based on minimum cost network flow, and the other, that is in general more efficient still, is based on a maximum flow presentation of the problem.  With those flow algorithms the problem is easily solved, in theory and in practice.   We further provide here a proof of NP-hardness of the problem on three covariates in Appendix \ref{sec:NPC}, that implies that no efficient algorithm exists for the minimum imbalance problem in the presence of three or more covariates.

\begin{appendices}

\section{NP-hardness for $P\geq3$}\label{sec:NPC}

Bennett et al. \cite{BVZ20} states that the min-imbalance problem is NP-hard  when $P\geq 3$ but presented no proof. Here, we provide a proof of this complexity status.  The implication of this result is that no efficient algorithm exists for the min-imbalance problem in the presence of three or more covariates.

The reduction is from the {\em 3-Dimensional Matching} problem, one of Karp's 21 NP-complete problems \cite{KARP72}\cite{GJ79}.
\\
{\em 3-Dimensional Matching}: Given a finite set $X$ and a set of triplets $U\subset X\times X\times X$. Is there a subset $M\subseteq U$ such that $|M|=|X|$ and that no two elements of $M$ agree in any coordinate?

\begin{lemma}\label{lem:strongnp} The min-imbalance problem is NP-hard even when $P=3$.
\end{lemma}

\begin{proof}
	Given an instance of 3-dimensional matching problem, we define an instance of min-imbalance problem with $P=3$, $k(p)=|X|$ for $p=1,2,3$, and $\ell_{p,k}=1$ for each $p,k$. For each triplet $u$ in $U$, we have a corresponding control sample $c$ in $\cC$ (the control group) such that the 3 covariates' levels of $c$ are defined by $u$.
	The decision problem for min-imbalance is stated as: is there a subset $S\subseteq \cC$ such that $|\cC_{p,k}\cap S |=\ell_{p,k}$ for all $p$ and $k$?
	
	We observe that if there exists a subset $M\subseteq U$ such that $|M|=|X|$ and that no two elements of $M$ agree in any coordinate, then on each coordinate, each element of $X$ must be covered by $M$ exactly once. That means, if we takes $S$ to be the subset of corresponding samples of $M$, then $|\cC_{p,k}\cap S |=1=\ell_{p,k}$ for all $p$ and $k$.
	
	And if there exists a subset $S\subseteq \cC$ such that $|\cC_{p,k}\cap S |=\ell_{p,k}=1$ for all $p$ and $k$, then the subset $M$ that corresponds to $S$ does not have two elements that agree in any coordinate, and that $|M|=|S|=|X|$.
	
	Therefore, 3-dimensional matching problem can be reduced to a 3-covariate min-imbalance problem, and hence, the min-imbalance problem is NP-hard even when $P=3$.
\end{proof}

\section{Selection of size different from $n$}\label{sec:sizeq}
The complexity of the min-imbalance problem on two covariates where the required size of the selection is $q<n'$, for $q$ different from the size of the treatment group $n$, is the same as for the case where the required size is $n$.  Both our methods described in Section \ref{sec:MCNF} and Section \ref{sec:max-flow} apply with minor modifications as described next.

\noindent
{\bf The minimum cost network flow method}
Consider the network flow formulation (\ref{MI}). With the notation defined in Section \ref{sec:ip}, $\sum_{i=1}^{k_1} \left(e_{1,i}-d_{1,i}\right)  = \sum_{i=1}^{k_1} \left(|S\cap L'_{1,i}|  - \ell_{1,i} \right) = |S| - n$. Therefore, if $|S|$ is required to be $q$, constraints (\ref{cons:sum1}) are changed to
\[ - \sum_{i_1=1}^{k_1} d_{1,i_1} + \sum_{i_1=1}^{k_1} e_{1,i_1} = q-n. \]
Similarly, constraints (\ref{cons:sum2}) are modified to
\[  \sum_{i_2=1}^{k_2} d_{2,i_2} - \sum_{i_2=1}^{k_2} e_{2,i_2} = n-q. \]
In Figure \ref{fig:mcnf}, the supply of node $1$ changes from $0$ to $q-n$, and the supply of node $2$ changes from $0$ to $n-q$. Everything else remains the same.  With these modifications it is easy to see that the optimal flow corresponds to an optimal solution to the min-imbalance with a selection of size $q$

\noindent
{\bf The maximum flow method}
The maximum flow problem used for $q \neq n$ is the same as for the case of $n$. There is however a slight modification in recovering the optimal selection from the max-flow solution.

If the size $q$ is less than the maximum flow amount $f^*$, then after the initialization step in Algorithm \ref{alg}, any $q$ samples out of the $S'$ selection, is an optimal selection. This is because each sample selected in the initialization step would reduce the total deficit by $2$ units and maintain the total excess to be zero, which is the best scenario possible.

If $q\geq f^*$, then Algorithm \ref{alg} with $n$ being replaced by $q$, gives the optimal selection. This follows from the proofs in Section \ref{sec:max-flow} with $n$ being replaced by $q$.

\section{The minimum $\kappa$-imbalance and $\kappa$-Matching-Balance problems}\label{sec:kappa}

The minimum $\kappa$-imbalance problem is to find a selection $S$ of control samples of size $\kappa n$, where $\kappa$ is an integer between $1$ and $\frac{n'}{n}$, and $n$ is the size of the treatment group, so as to minimize the $\kappa$-imbalance, defined as $\sum_{p=1}^P \sum_{i=1}^{k_p} | {|S\cap L'_{p,i}|}-\kappa \ell_{p,i} |$.   For $\kappa =1$ the problem is the min-imbalance problem and therefore it appears that the  $\kappa$-imbalance problem is more general than the min-imbalance problem.  We show here however that the minimum $\kappa$-imbalance problem is reducible to the min-imbalance problem and therefore the two problems are equivalent.


To reduce the minimum $\kappa$-imbalance problem to the min-imbalance problem we replace each treatment sample by $\kappa$ copies. This new treatment group has a size of $\kappa n$ and the size of level $i$ under covariate $p$ is $\kappa \ell_{p,i}$ for each $p$ and $i$. Hence the total imbalance defined for the new treatment group is the same as the $\kappa$-imbalance of the original treatment group.
This implies that particularly for the two covariates minimum $\kappa$-imbalance problem the methods in Section \ref{sec:MCNF} and Section \ref{sec:max-flow} provide an optimal solution in polynomial time.  If this solution, given in terms of the sizes of the level-intersections (the numbers of control samples selected in each  $L'_{1,i_1} \cap  L'_{2,i_2}$ for $i_1=1, \ldots k_1$ and $i_2 =1 ,\ldots ,k_2$), is unique, then the solution to the second stage of MB solves that problem optimally and in polynomial time.

For any number of covariates $P$ an optimal solution to the second stage of MB, for given sizes of level intersections $s_{i_1,i_2,\ldots ,i_P}$, is presented as the following minimum cost flow problem.  First there is a bipartite graph with the treatment samples each represented by a node on one side, and the control samples each represented by a node on the other size.  The cost on each arc between a treatment sample and a control sample is the ``distance" metric between the two, and the arc capacity is $1$.  Each treatment sample has a demand of $\kappa$.    For each non-zero level intersection there is a source node with supply of $s_{i_1,i_2,\ldots ,i_P}$.   This source node is connected to all control samples in the intersection $L'_{1,i_1} \cap  L'_{2,i_2}\cap \ldots \cap L'_{i_P}$ with arcs of capacity $1$ and cost of $0$.   The control sample nodes through which there is a positive flow (of $1$ unit) are the ones selected and assigned to the respective treatment sample nodes to which they have a positive flow.  This is a minimum cost network flow problem with a total demand (or supply) bounded by $n'$, and $O(n')$ arcs and $O(n')$ nodes.  Therefore the successive shortest paths algorithm, discussed in Section \ref{sec:MCNF}, solves this problem in $O(n'( n'+ n' \log n')$ steps.

\section{The minimum cost network flow, the maximum flow problems and total unimodularity}\label{sec:flow}
We formulate here the {\em minimum cost network flow problem} (MCNF). The input to the problem is a graph $G=(V,A)$ with a set of nodes $V$ and a set of arcs $A$, where each arc $(i,j)\in A$ is associated with a cost $c_{ij}$, capacity upper bound $u_{ij}$, and capacity lower bound $l_{ij}$. Each node $i\in V$ has supply $b_i$ which is interpreted as demand if negative, and can be $0$.   Let $x_{ij}$ be the amount of flow on arc $(i,j)\in A$.  The flow vector $\x$ is said to be feasible if it satisfies:\\
(1) {\bf Flow balance constraints:}  For every node $k \in V$  $Outflow(k)-Inflow (k) = b_k$\\ 
(2) {\bf Capacity constraints:}  For each arc $(i,j)\in A$, $l_{ij} \leq x_{ij} \leq u_{ij}$.

The linear programming formulation of the problem is:
\[
 \hspace{.4in}\begin{array}{ll}
 \mbox{(MCNF)~~~~}  \min & \sum_{(i,j)\in A}\ c_{ij} x_{ij}  \\
\mbox{subject to }\ & \sum_{(k,j)\in A} x_{kj}\,-\,
                         \sum_{(i,k)\in A} x_{ik} = b_i  \ \forall k \in V \\
&     l_{ij} \le x_{ij} \le u_{ij}, \  \ \forall (i,j) \in A.
\end{array}
\]

\begin{theorem}\label{thm:TU}\cite{HK56}\cite{Heller56}
	A $\{ 0,1,-1\}$-matrix where each column (or row) has at most one $1$ and at most one $-1$ is totally unimodular.
\end{theorem}

The addition of the capacity constraints retains the total unimodularity property.  Note that the MCNF constraint matrix is linearly dependent (with the sum of all balance constraints leading to $0=0$) and one constraint can be eliminated as it is the negative sum of the others.  

We also make use of the maximum flow problem, {\em max-flow}, which is a special case of MCNF.  The max-flow problem is defined on a directed graph $G=(V,A)$, with arc capacities $u_{ij}$ and two distinguished nodes: a source node, $s\in V$, and a sink node, $t\in V$.  The total outflow from $s$, or the total inflow into $t$, is called the value of the flow. The objective is to find the maximum value feasible flow leaving $s$ and reaching $t$ that satisfy the arc capacities.  A linear programming formulation of max-flow is:
\[
 \hspace{.4in}\begin{array}{ll}
 \mbox{(max-flow)~~~~}  \max & f  \\
\mbox{subject to }\ & \sum_{(s,j)\in A} x_{sj} = f \\
& \sum_{(k,j)\in A} x_{kj}\,-\,
                         \sum_{(i,k)\in A} x_{ik} = 0  \ \forall k \in V \setminus \{ s,t \}\\
&   0 \le x_{ij} \le u_{ij}, \  \ \forall (i,j) \in A.
\end{array}
\]

As a special case of MCNF there are specialized algorithms for max-flow which are more efficient than algorithms for MCNF.

\end{appendices}


\begin{thebibliography}{10}
\expandafter\ifx\csname url\endcsname\relax
  \def\url#1{\texttt{#1}}\fi
\expandafter\ifx\csname urlprefix\endcsname\relax\def\urlprefix{URL }\fi
\expandafter\ifx\csname href\endcsname\relax
  \def\href#1#2{#2} \def\path#1{#1}\fi

\bibitem{Rosenbaum07}
P.~R. Rosenbaum, R.~N. Ross, J.~H. Silber, Minimum distance matched sampling
  with fine balance in an observational study of treatment for ovarian cancer,
  Journal of the American Statistical Association 102~(477) (2007) 75--83.
\newblock \href {https://doi.org/10.1198/016214506000001059}
  {\path{doi:10.1198/016214506000001059}}.

\bibitem{Yang12}
D.~Yang, D.~S. Small, J.~H. Silber, P.~R. Rosenbaum, Optimal matching with
  minimal deviation from fine balance in a study of obesity and surgical
  outcomes, Biometrics 68~(2) (2012) 628--636.
\newblock \href {https://doi.org/10.1111/j.1541-0420.2011.01691.x}
  {\path{doi:10.1111/j.1541-0420.2011.01691.x}}.

\bibitem{BVZ20}
M.~Bennett, J.~P. Vielma, J.~R. Zubizarreta, Building representative matched
  samples with multi-valued treatments in large observational studies, Journal
  of Computational and Graphical Statistics 0~(0) (2020) 1--29.
\newblock \href {https://doi.org/10.1080/10618600.2020.1753532}
  {\path{doi:10.1080/10618600.2020.1753532}}.

\bibitem{Jewell58}
W.~S. Jewell, Optimal flow through networks, in: Operations Research, Vol.~6,
  1958, pp. 633--633.

\bibitem{Iri60}
M.~Iri, A new method of solving transportation-network problems, Journal of the
  Operations Research Society of Japan 3~(1) (1960) 2.

\bibitem{BG61}
R.~Busaker, P.~J. Gowen, A procedure for determining minimal-cost network flow
  patterns, Tech. rep., ORO Technical Report 15, Operational Research Office,
  John Hopkins University (1961).

\bibitem{Tomizawa71}
N.~Tomizawa, On some techniques useful for solution of transportation network
  problems, Networks 1~(2) (1971) 173--194.
\newblock \href {https://doi.org/10.1002/net.3230010206}
  {\path{doi:10.1002/net.3230010206}}.

\bibitem{EK72}
J.~Edmonds, R.~M. Karp, Theoretical improvements in algorithmic efficiency for
  network flow problems, Journal of the ACM (JACM) 19~(2) (1972) 248--264.
\newblock \href {https://doi.org/10.1145/321694.321699}
  {\path{doi:10.1145/321694.321699}}.

\bibitem{GR98}
A.~V. Goldberg, S.~Rao, Beyond the flow decomposition barrier, Journal of the
  ACM (JACM) 45~(5) (1998) 783--797.

\bibitem{KARP72}
R.~M. Karp, Reducibility among combinatorial problems, in: Complexity of
  computer computations, Springer, 1972, pp. 85--103.

\bibitem{GJ79}
M.~R. Garey, D.~S. Johnson, Computers and Intractability: A Guide to the Theory
  of NP-Completeness, W. H. Freeman, New York, NY, USA, 1978.

\bibitem{HK56}
I.~Heller, C.~Tompkins, Integral boundary points of convex polyhedra, Linear
  inequalities and related systems 38 (1956) 223--246.

\bibitem{Heller56}
I.~Heller, C.~Tompkins, An extension of a theorem of dantzig’s, Linear
  inequalities and related systems 38 (1956) 247--254.

\end{thebibliography}

\end{document}